\documentclass{amsart}

\usepackage{enumerate}
\usepackage{amssymb}
\usepackage{amsfonts}
\usepackage{latexsym}
\usepackage{amsmath}
\usepackage{euscript}
\usepackage{graphics}
\usepackage{graphicx}
\usepackage{tikz}
\usetikzlibrary{spy}
\graphicspath{ {./images/} }

\usepackage{dirtytalk}

\newtheorem{theorem}{Theorem}[section]
\newtheorem{proposition}[theorem]{Proposition}

\newtheorem{lemma}[theorem]{Lemma}
\newtheorem{definition}[theorem]{Definition}
\theoremstyle{remark}
\newtheorem{example}[theorem]{Example}
\newtheorem{remark}[theorem]{Remark}

\numberwithin{equation}{section}

\usepackage[colorinlistoftodos]{todonotes}

\newmuskip\pFqmuskip

\newcommand*\pFq[6][8]{%
  \begingroup % only local assignments
  \pFqmuskip=#1mu\relax
  \mathcode`\,=\string"8000
  \begingroup\lccode`\~=`\,
  \lowercase{\endgroup\let~}\pFqcomma
  {}_{#2}F_{#3}{\left[\genfrac..{0pt}{}{#4}{#5};#6\right]}%
  \endgroup
}
\newcommand{\pFqcomma}{\mskip\pFqmuskip}

\begin{document}
\title[Orthogonal Polynomials and Perfect State Transfer]{Orthogonal Polynomials and Perfect State Transfer}

\author[R.~Bailey]{R.~Bailey}
\address{
RB,
Department of Mathematical Sciences\\
Bentley University\\
175 Forest Street\\
Waltham, MA 02452, USA}
\email{rbailey@bentley.edu}

%\author[M.~Nathanson]{Michael~Nathanson}
%\address{
%MN,
%Department of Mathematical Sciences\\
%Bentley University\\
%175 Forest Street\\
%Waltham, MA 02452, USA}
%\email{mnathanson@bentle%y.edu}


\begin{abstract}
The aim of this review paper is to discuss some applications of orthogonal polynomials in quantum information processing. The hope is to keep the paper self contained so that someone wanting a brief introduction to the theory of orthogonal polynomials and continuous time quantum walks on graphs may find it in one place. In particular, we focus on the associated Jacobi operators and discuss how these can be used to detect perfect state transfer. We  also discuss how orthogonal polynomials have been used to give results which are analogous to those given by Karlin and McGregor when studying classical birth and death processes. Finally, we show how these ideas have been extended to quantum walks with more than nearest neighbor interactions using exceptional orthogonal polynomials. We also provide a (non-exhaustive) list of related open questions.
\end{abstract}
\maketitle

\section{Introduction}

The pursuit of quantum computation has inspired much work both in applied physics and in mathematics. Quantum computers process information in a fundamentally different way than any classical device, and they promise an exponential speedup over the classical computers for solving certain types of problems, aiding in the development cryptography, machine learning and AI. They could also prove useful for modeling quantum physical systems. 
Understanding and managing the interactions of quantum particles is a key challenge of implementing quantum computation, and one can model the existence of physical interactions between particles as edges connecting vertices in a graph. This paper will focus on path graphs representing quantum spin chains, i.e. linear arrangements of quantum particles. Spin chains serve as simplified representations of more complex systems and thus we can gain a lot of insight into quantum behavior by studying these models. We will study an extension to a larger class of graphs in Section \ref{sec:XOPS}.

To begin, recall that a quantum bit (qubit) is the fundamental unit of information in quantum computing, similar to an ordinary bit except that where a bit takes on the value 0 or 1, a qubit can be in a superposition of both 0 and 1. We only can determine the probability that after observation, the qubit will be 0 or 1. The \textit{state} of a qubit is then represented by a vector $v=(c_0,c_1)^{\top} \in \mathbb{C}^2$ where $|v|=1$. Then $|c_0|^2$ is the probability that after measuring, the qubit is 0, and $|c_1|^2$ is the probability the qubit is 1. Putting $n$ qubits together, the state of a spin chain with $n$ qubits is likewise a unit vector in $\mathbb{C}^{2^n}$ that encodes how the qubit systems are oriented. 

%Quantum computers use the principles of quantum mechanics to process information in ways that classical computers cannot. The goal is to increase computational power in order to solve problems in a fraction of the time, aiding in the development of areas such as cryptography, machine learning and artificial intelligence. Studying how quantum particles interact is crucial to understanding quantum computing and thus mathematicians utilize graphs to represent these interactions. This paper will focus on path graphs representing quantum spin chains; linear arrangements of quantum particles. Spin chains serve as simplified representations of more complex systems and thus we can gain a lot of insight into quantum behavior by studying these models. 

%To begin, recall that a quantum bit (qubit) is the fundamental unit of information in quantum computing, similar to an ordinary bit except that where a bit takes on the value 0 or 1, a qubit can be in a superposition of both 0 and 1. We only can determine the probability that after observation, the qubit will be 0 or 1. The \textit{state} of a qubit is then represented by a vector $v=(c_0,c_1)^{\top} \in \mathbb{C}^2$ where $|v|=1$. Then $|c_0|^2$ is the probability that after measuring, the qubit is 0, and $|c_1|^2$ is the probability the qubit is 1. The state of a spin chain is a description of how the qubits are oriented and how they interact with one another. 

To study the evolution of a quantum system over time, we consider the evolution operator, $U(t)=e^{itH}$, where $H$ is the \textit{Hamiltonian} of the system (a unitary operator representing the total energy of the system). For continuous time quantum walks on graphs, the Hamiltonian is taken to be the adjacency matrix or the graph Laplacian. We provide more detail in Section \ref{sec:PST}.

The subject of transferring a quantum state from one location to another within a quantum computer is an important task and has gained much interest from the mathematical community (see \cite{Godsil2017}, \cite{KAY_2010}, \cite{Kirkland2019}...) In particular, it is desirable to say with probability one, that the state of one qubit is transferred to another after some time. This phenomenon is called \textit{perfect state transfer} (PST) and in what follows, we will highlight how the theory of orthogonal polynomials has played an important role in the study of PST. We begin by reviewing important properties of orthogonal polynomials and Jacobi matrices.

%In Section \ref{sec:OPS}, we review some basic definitions and theorems from the theory of orthogonal polynomials. In Section \ref{sec:finiteOPS}, we connect orthogonal polynomials to tridiagonal matrices and provide an explicit example. Section \ref{sec:PST} then defines perfect state transfer (PST) and discusses some examples where the theory of orthogonal polynomials has been used to study PST. In Section \ref{sec:karlinmcgregor}, we recall the famous integral theorem by Karlin and McGregor for transition probabilities of birth and death processes and show how this has been extended to continuous time quantum walks. We discuss exceptional orthogonal polynomials in Section \ref{sec:XOPS} and explore how they can be used to study quantum walks on graphs with non-typical behaviors. Finally, we conclude with a few open questions related to perfect state transfer and quantum walks on graphs.

\section{Orthogonal Polynomials}\label{sec:OPS}
\begin{definition}
Let $\mathcal{L}$ be a complex valued linear functional on the vector space of all polynomials  $\mathbb{P}$. A sequence of polynomials, $\{P_{n}(x)\}_{n=0}^{\infty}$, is an \textbf{orthogonal polynomial sequence} (OPS) with respect to $\mathcal{L}$ provided for all $m, n=0,1,2\dots$
\begin{enumerate}

    \item $P_n(x)$ is a polynomial of degree $n$,
    \item $\mathcal{L}[P_n(x)P_{m}(x
    )]=0$ for $n \neq m$,
    \item $\mathcal{L}[P_n^2(x)]\neq 0$.   
\end{enumerate}
\end{definition}
Often, the linear functional $\mathcal{L}$ takes the form
\[
\mathcal{L}[p]=\int_{a}^bp(x)w(x)\, dx
\] where $a$ and $b$ may be infinite and $w(x)$ is a nonnegative function which is positive on a subset of $(a,b)$ of positive Lebesgue measure. The function $w(x)$ is called the \textit{weight function}.

\iffalse
\begin{example}(Chebyshev Polynomials)
One of the most well-known OPSs is the Chebyshev polynomials $\{T_n(x)\}_{n=0}^{\infty}$. The Chebyshev polynomials of the first kind are defined by
\[
T_n(x)=\cos(n\theta)=\cos(n \cos^{-1}(x)), \quad -1\leq x \leq 1.
\]
and are eigenfunctions to the following second-order differential operator
\[
L[y]:=(1-x^2)\frac{d^2y}{dx^2}-x\frac{dy}{dx}
\] with eigenvalues $\lambda_n= -n^2$.

They are orthogonal with respect to the following linear functional
\[
\mathcal{L}[p]=\int_{-1}^1p(x)\frac{1}{\sqrt{1-x^2}}\,dx.
\]
\end{example}
\fi
Because the purpose of this paper is to study applications of OPS in continuous time quantum walks on finite graphs, we will focus on finite sequences of orthogonal polynomials, which are families of  polynomials $\{P_n(x)\}_{n=0}^N$ satisfying the above properties for all $n=0,1,2,\dots, N$. 

A well known example of finite orthogonal polynomials are the Krawtchouk polynomials, introduced by Mykhailo Kravchuk when studying a discrete analogue of the Hermite polynomials (an OPS with respect to the Gaussian density function).
%\begin{example}(Krawtchouk Polynomials)
\begin{definition}\label{def:krawtchouk}
    Let $0<p<1$ and let $M\in \mathbb{Z}_+$. Then for $n=0,1,2,\dots,M$, the $n$-th degree monic Krawtchouk polynomial is defined by
\begin{equation}\label{def:Krawtchouk}
\begin{split}
    K_n(x;p,M) &= (-M)_np^n
\pFq{2}{1}{-x,-n}{-M}{1/p}\\
&=\sum_{j=0}^n\frac{(-n)_j(-M+j)_{n-j}}{j!}p^{n-j}(-x)_j
\end{split}
\end{equation} where $(a)_n$ is the Pochhammer symbol 
   \[(a)_n =\begin{cases}
   1& (n=0)\\
   a(a+1)\dots (a+n-1) & (n=1,2,\dots)
   \end{cases}.
   \]
\end{definition}
The Krawtchouk polynomials are orthogonal with respect to the following linear functional,
\begin{equation}
   \mathcal{L}[p]= \sum_{x=0}^M\binom{M}{x}p^x(1-p)^{M-x}p(x).
\end{equation}

\noindent Below are the first few monic Krawtchouk polynomials for $M=4$ and $p=1/2$.
\begin{align*}
    K_0(x)&=1\\
    K_1(x)&=x-2\\
    K_2(x)&=x^2-4x+3\\
    K_3(x)&=x^3-6x^2+\frac{19}{2}x-3\\
    K_4(x)&=x^4-8x^3+20x^2-16x+\frac{3}{2}
\end{align*}
%\end{example}
It is evident from their weight function that the Krawtchouk polynomials have many applications in probability, and in particular, are linked to the binomial distribution and random walks. We will exploit this application in Section 
\ref{sec:karlinmcgregor} and discuss how it has been used to study perfect state transfer in quantum walks on path graphs.

\section{Finite Orthogonal Polynomials and Jacobi Matrices}\label{sec:finiteOPS}

\noindent Given an $N\times N$ Jacobi matrix $J$,
\begin{equation}\label{eq:normalized}
J=\begin{pmatrix}
b_0 & a_0 & 0 &\cdots &0&0   \\
     a_0 & b_1 & a_1 & \cdots&0 &0 \\
     0 & a_1 & b_2 & \cdots&0&0  \\
      \vdots& \vdots& \vdots&\ddots &\vdots&\vdots\\
      0&0&0&\cdots&b_{N-2}&a_{N-2}\\
      0&0&0&\cdots&a_{N-2}&b_{N-1}
\end{pmatrix}
\end{equation} with $a_n\neq 0, n \geq 0$, one may associate to it a family of orthonormal polynomials, $\{P_n(x)\}_{n=0}^{N}$, defined via the following three-term recurrence relation:
\begin{equation}\label{SymTriTermNSec5}
a_nP_{n+1}(x)+b_nP_n(x)+a_{n-1}P_{n-1}(x)=xP_n(x), \quad n=0,1,2,\dots,N-1
\end{equation} with $P_{-1}(x):=0$ and $P_{0}(x):=1$ (see Favard's Theorem, \cite{Chihara}). Here, $a_{N-1}$ can be chosen to be an arbitrary constant so we take it to be 1. 
%When $n=N-1$, the relation becomes
%\begin{equation}\label{eq:P_Nrec}
%xP_{N-1}(x)=\alpha_{N-1}\prod_{k=0}^{N-1}(x-\lambda_k)+b_{N-1}P_{N-1}(x)+a_{N-2}P_{N-2}(x)
%\end{equation} where $\{\lambda_k\}_{k=0}^{N-1}$ are the eigenvalues of $H$ and $\alpha_{N-1}$ is a nonzero constant. From equation \eqref{eq:P_Nrec}, one can see that $P_N(x)=\alpha_{N-1}(x-\lambda_0)(x-\lambda_1)\dots(x-\lambda_{N-1})$. We will refer to the family $\{P_n(x)\}_{n=0}^N$ as the \textit{OPS generated by $H$}.

If $a_n>0$ for all $n=0,1,\dots, a_{N-2}$, then the linear functional corresponding to $\{P_n(x)\}_{n=0}^{N}$ is positive-definite on span$(1,x,x^2,\dots, x^{N-1})$ and takes the form
\[\mathcal{L}[p]=
\sum_{k=0}^{N-1}p(\lambda_k)w(\lambda_k)
\] for a weight function $w(x)$ which is positive on the $\lambda_k$. In this case, we obtain the following important properties of the zeros of $\{P_n(x)\}_{n=0}^{N}$ (one may refer to \cite{Baik2007} for more information on discrete orthogonal polynomials).

\begin{proposition}\label{prop:zeros}
The zeros of the $\{P_n(x)\}_{n=0}^{N}$ are real and simple.
\end{proposition}

\begin{proposition}\label{prop:interlacezeros}
Let $x_{n,k}$ denote the $k$-th zero of $P_{n}(x)$. Then
the zeros of $P_n(x)$ and $P_{n+1}(x)$ interlace, i.e.
\[
x_{n+1,1}<x_{n,1}<x_{n+1,2}<\dots<x_{n,n}<x_{n+1,n+1}
\] for all $n=0,1,\dots, N-1$.
\end{proposition}

\begin{remark}
    We note that Favard's theorem, Proposition \ref{prop:zeros} and Proposition \ref{prop:interlacezeros} also hold in the case $J$ is a semi-infinite Jacobi matrix, but we restrict ourselves to the finite case due to the nature of the application in quantum computing which we review.
\end{remark}
As a consequence of the uniqueness of orthogonal polynomial sequences, it is often convenient to work with the monic polynomials $p_n(x):=a_{-1}a_0a_1\dots a_{n-1}P_n(x)$ where $a_{-1}:=1$. These polynomials satisfy
\begin{equation}\label{eq:monicrec}
xp_n(x)=p_{n+1}(x)+b_np_n(x)+a^2_{n-1}p_{n-1}(x), \quad n=0, 1, 2,\dots, N-1
\end{equation} which corresponds to the following monic Jacobi matrix,
\[
\mathcal{J}=\begin{pmatrix}
b_0 & 1&0 & 0&\cdots   \\
     a^2_0 & b_1 & 1 &0&   \\
     0 & a^2_1 & b_2 & 1&\ddots  \\
      \vdots& & \ddots & \vdots&1\\
      0&0&0&a^2_{N-2}&b_{N-1}
\end{pmatrix}.
\] Let $J_n$ be the truncated $n\times n$ upper left sub-matrix of $\mathcal{J}$ and $I_n$ be the $n\times n$ identity matrix. Define polynomials of degree $n$ by $q_n(x)=\det(xI_n-J_n)$. Then the $q_n(x)$ satisfy \eqref{eq:monicrec} with $q_1(x)=x-b_0=p_1(x)$ and $q_2(x)=(x-b_0)(x-b_1)-a_0^2=p_2(x)$, hence we must have that $q_n(x)=p_n(x)$ for all $n=0,1,2,\dots, N$. Therefore, the OPS corresponding to $\mathcal{J}$ may also be defined by
\[
p_n(x)=\det(xI_n-J_n), \quad n=1,2,\dots,N.
\]
%The $\{p_k(x)\}_{k=0}^N$ also take the following form
%\[
%p_k(x)=\det(xI_k-J_k), \quad k=1,2,\dots,N
%\] where $p_0:=1$, $J_k$ is the truncated $k\times k$ upper left sub-matrix of $J$ and $I_k$ is the $k\times k$ identity matrix (this can be seen since defining polynomials in such a way results in polynomials satisfying \eqref{eq:monicrec} with $p_1(x)=x-b_0$ and $p_2(x)=(x-b_0)(x-b_1)-a_0^2$).\\
Using this view, it is clear that $p_N(x)$ is the characteristic polynomial of $\mathcal{J}$ and thus it's zeros (and hence the zeros of $P_n(x)$) are the eigenvalues of $\mathcal{J}$, which must be simple and real by Proposition \ref{prop:zeros}. 
\begin{remark}We remark here that the realness of the eigenvalues of $\mathcal{J}$ also follows from the fact that $\mathcal{J}$ is similar to $J$ and $J$ is Hermitian. Thus, $J$ (and hence $\mathcal{J}$) has real eigenvalues. However, the theory of orthogonal polynomials gives a quick proof as to the multiplicity of the eigenvalues.\end{remark}
It can also be seen that the eigenvectors of $J$ are given in terms of the corresponding orthogonal polynomials in the following way. Let $\{\lambda_k\}_{k=0}^{N-1}$ be the zeros of $P_N(x)$. Then \eqref{eq:monicrec} gives 
\begin{equation}\label{eq:recmatrix}
\begin{pmatrix}
b_0 & a_0 & 0 &\cdots &0&0   \\
     a_0 & b_1 & a_1 & \cdots&0 &0 \\
     0 & a_1 & b_2 & \cdots&0&0  \\
      \vdots& \vdots& \vdots&\ddots &\vdots&\vdots\\
      0&0&0&\cdots&b_{N-2}&a_{N-2}\\
      0&0&0&\cdots&a_{N-2}&b_{N-1}
\end{pmatrix}\begin{pmatrix}P_0(\lambda_k)\\
    P_1(\lambda_k)\\
    P_2(\lambda_k)\\
    \vdots\\
    P_{N-2}(\lambda_k)\\
    P_{N-1}(\lambda_k)
  
\end{pmatrix}+\begin{pmatrix}
    0\\
    0\\
    0\\
    \vdots\\
    0\\
    P_N(\lambda_k)
      \end{pmatrix}=\lambda_k\begin{pmatrix}P_0(\lambda_k)\\
    P_1(\lambda_k)\\
    P_2(\lambda_k)\\
    \vdots\\
    P_{N-2}(\lambda_k)\\
    P_{N-1}(\lambda_k)
  
\end{pmatrix}
\end{equation}
where $P_N(\lambda_k)=0$, thus equation \eqref{eq:recmatrix} is equivalent to
\begin{equation}\label{eq:evector}
Jv_k=\lambda_kv_k \text{ where }  v_k=[P_0(\lambda_k), P_1(\lambda_k),\dots,P_{N-1}(\lambda_k)]^{\top}
\end{equation} where $v^{\top}$ denotes the transpose of $v$. We highlight here that $v_k$ is indeed an eigenvector since by Proposition \ref{prop:interlacezeros}, $P_{N-1}(\lambda_k)\neq 0$ for any $k=0,1,2,\dots,N-1$. %Note the monic case is analogous since $P_N(x) $ and $p_N(x)$ have the same zeros.

%Note that $\{p_0(x), p_1(x),\dots,p_{n-1}(x)\}$ forms a basis for the vector space of polynomials of degree at most $n-1$. Thus, using \eqref{eq:monicrec}, we may consider $J_m$ as the operator of multiplication by $x\mod{p_{n}(x)}$ in this space.

\begin{example}\label{ex:krawtchouk}(Krawtchouk Jacobi Matrix) We recall the monic Krawtchouk polynomials given in Definition \ref{def:krawtchouk}.
  They satisfy the following three-term recurrence relation
\begin{equation}\label{eq:krawtchoukmonicrec}
   xK_n(x;p,M)=K_{n+1}(x;p,M)+\frac{M}{2}K_n(x;p,M)+\frac{(M+1-n)n}{4}K_{n-1}(x;p,M).
   \end{equation}
For example, when $M=4$ and $p=1/2$, the associated monic Jacobi matrix is  
   \[
\mathcal{J}=\begin{pmatrix}
2 & 1&0 & 0&0   \\
     1 & 2 & 1 &0&0   \\
     0 & 3/2 & 2 & 1&0  \\
      0&0& 3/2 & 2&1\\
      0&0&0&1&2
\end{pmatrix}.
\]
Then $K_5(x;1/2,4)=\det(xI-\mathcal{J})=-x^5+10x^4-35x^3+50x^2-24x$ which has zeros at $x=0,1,2,3,4$ hence the eigenvalues of $J$ are $0,1,2,3$ and $4$ with corresponding eigenvectors $v_k=(K_0(k), K_1(k), K_2(k), K_3(k), K_4(k))^{\top}$, $k=0,1,2,3,4$.

\end{example}

\section{Perfect State Transfer}\label{sec:PST}
Recall that at path graph $G$ with $N$ vertices is a connected, undirected graph where two end vertices have degree 1 and the remaining $N-2$ vertices have degree 2. This means that $G$ may be drawn so that its vertices all lie on a straight line and are only connected to their nearest neighbors. We can associate Jacobi matrices to path graphs in the following way. Label the vertices of $G$ by $0, 1, 2,\dots, N-1$, and let $w_{jk}$ denote the weight of the edge connecting vertex $j$ to vertex $k$. We will assume the weights satisfy $w_{jk}>0$ for $j\neq k$. Then the \textit{adjacency matrix} of $G$ is the $N\times N$ matrix defined by $J=(w_{jk})_{j,k=0}^{N-1}.$ One can see that since each vertex is connected to only its nearest neighbors, $J$ will be tridiagonal and hence a Jacobi matrix. In the case of path graphs, $G$ is undirected, hence $w_{jk}=w_{kj}$ and therefore $J$ is symmetric. 

It is natural to use graph theory to model quantum communication by letting the vertices represent qubits and the edge weights represent couplings (how one qubit interacts with another). In the case of path graphs, $J$ represents the Hamiltonian of a spin chain.

 Let ${\bf e}_n=(0,0,\dots,0,1,0,\dots, 0)^{\top}$ be the standard basis vector of $\mathbb{C}^N$ with 1 in the $n+1$-th entry and zeros elsewhere, $n=0, 1, 2, \dots, N-1$.
For example, ${\bf e}_0=(1,0,0,\dots,0)^{\top}$ and ${\bf e}_{N-1}=(0,0,\dots,0,1)^{\top}$.

\begin{definition}[Perfect State Transfer]\label{def:PST} Let $j$ and $k$ be vertices of a finite graph and $J$ be its adjacency matrix. Then \textit{perfect state transfer (PST)} from vertex $n$ to vertex $m$ occurs if there exists some time $t_0>0$ such that \[\left|{\bf e}^{\top}_m\exp(it_0J){\bf e}_n\right|^2=1.\] 
\end{definition} 
\begin{remark}
    PST between vertices $n$ and $m$ means that an initial state where only qubit $n$ is excited, evolves with probability 1 to a state where only qubit $m$ is excited.
\end{remark}
\begin{remark}
We also note that definition \ref{def:PST} is equivalent to 
\[
e^{it_0J}{\bf e}_n=e^{i\phi}{\bf e}_m
\]for some real number $\phi$.  
\end{remark}

Many results have been formed concerning quantum spin chains with $N$ qubits. It was shown in \cite{Christandl_2004}
that for unweighted paths where the Hamiltonian is represented by the adjacency matrix, PST between endpoints can only occur for $N \leq 3$. Therefore, considering weighted paths became of interest and in 2005, Christandl et al showed that it is possible to achieve PST between endpoints in this case by allowing couplings between qubits which correspond to adding specific weights to the edges of the path (see \cite{Christandl_2005}). Because the graph Hamiltonian for a path graph with positive weights is an $N\times N$ Jacobi matrix with positive off-diagonal entries, there is a deep connection between PST and the theory of orthogonal polynomials. Below, we provide several examples of where OPSs have proven to be useful for studying PST in spin chains.
%it is shown that if PST between endpoints of a symmetric spin chain occur, then the Hamiltonian must be mirror symmetric i.e symmetric about the antidiagonal. This implies that the corresponding eigenvectors of $H$ must be either symmetric, or antisymmetric and in fact, after ordering the eigenvalues $\lambda_n<\lambda_{n-1}$, it is shown that the number of sign changes in the corresponding eigenvector $|\lambda_n\rangle$ is $n-1$ and hence the ordered eigenvectors alternate between symmetric and antisymmetric. This observation was key to proving the following well known characterizing property of spin chains with PST.
%The following is a well-known lemma describing PST in terms of the associated family of orthogonal polynomials. (for instance, see \cite{Kirkland2019} or \cite{vinet2012}). 
\begin{lemma}\label{lemma:PST}
Let $G$ be a path graph on $N$ vertices, $J$ its weighted adjacency matrix and $\{P_n(x)\}_{n=0}^{N}$ be the OPS generated by $J$. Denote the zeros of $P_N(x)$ by $
\lambda_{k}$, $k=0,1, 2, \dots, N-1$.  If there is PST between vertices $n$ and $m$ at time $t_0$, then  \[
    P_m(\lambda_{k})=e^{-i\phi}e^{i t_0\lambda_{k}}P_{n}(\lambda_{k})
    \] for all $k=0,1,2,\dots, N-1$ and some phase factor $\phi \in \mathbb{R}$.
\end{lemma}
The proof utilizes the fact that $J$ can be diagonalized as $J=PDP^{\top}$, where \newline $D=$ diag$\{\lambda_0,\lambda_1,\dots,\lambda_{N-1}\}$ is a diagonal matrix whose entries are the zeros of $P_N(x)$, and $P$ is an $N\times N$ matrix whose columns are orthonormal eigenvectors $\tilde{v}_k=\frac{v_k}{\|v_k\|}$ corresponding to $\lambda_k$, where the $v_k$ are as in \eqref{eq:evector} (see \cite{Kirkland2019} or \cite{vinet2012}).

%\begin{proof}
%Assume PST is realized between vertices $n$ and $m$ at time $t_0$. Then
%\begin{equation}\label{eq:PST}
%e^{it_0H}e_{n}=e^{i\phi}e_{m}
%\end{equation} for some  $\varphi \in \mathbb{R}$. Since the eigenvalues of $H$ are the zeros of $P_N(x)$, they are real and simple and the corresponding eigenvectors $\{v_k\}_{k=0}^{N-1}$ form a basis for $\mathbb{R}^N$. Recall these eigenvectors may be written in terms of the corresponding orthogonal polynomials $v_k=[P_0(\lambda_k), P_1(\lambda_k),\dots,P_{N-1}(\lambda_k)]^{\top}$. Let $\tilde{v_k}=c_kv_k$ be the normalized eigenvectors so that $\langle \tilde{v_k},\tilde{v_j}\rangle =\delta_{kj}$. Then, we may expand $e_n$  in terms of the $\tilde{v_k}$, so that
%\[
%e_n=\sum_{k=0}^{N-1}\alpha_kP_n(\lambda_k)\tilde{v_k}.
%\] Using this expansion and the fact that $V=[\tilde{v_0}, \tilde{v_1}, \dots, \tilde{v_{N-1}}]$ diagonalizes $H$, we have that \eqref{eq:PST} is equivalent to 
%\begin{equation}\label{eq:polyPST}
%P_{m}(\lambda_k)=e^{-i\phi}e^{it_0\lambda_k}P_n(\lambda_k),\quad k=0, 1,2,\dots, N-1.
%\end{equation} 
%\end{proof}
\begin{remark}\label{remark:pstendpoints}
In the case of PST between endpoints, the converse of the Lemma \ref{lemma:PST} is also true and thus PST occurs between nodes 0 and $N-1$ at time $t_0$ if and only if $P_{N-1}(\lambda_k)=e^{-i\phi}e^{it_0\lambda_k}$ (see Theorem 6.1 of \cite{DEREVYAGIN21} for a detailed proof).

\end{remark}
The following well known theorem translates the question of PST into a question about the eigenvalues of $J$.

\begin{theorem} Let $J$ be the weighted adjacency matrix of a path graph on $N$ vertices and let
$\{\lambda_k\}_{k=0}^{N-1}$ denote the (distinct) eigenvalues of $J$. If $J$ achieves PST at time $t_0$ then the eigenvalues satisfy $\lambda_k-\lambda_{k-1}=(2m_k+1)\pi/t_0$ where $m_k$ is a nonegative integer.
\end{theorem}
As in Lemma \ref{lemma:PST}, the proof uses the fact that the eigenvalues and eigenvectors of $J$ can be given in terms of the corresponding orthogonal polynomials.
It also uses the interlacing property of the zeros of $P_n(x)$ and $P_{n+1}(x)$ which is illustrated below. \begin{proof}
Assume PST is realized at time $t_0$. Then by 
Remark \ref{remark:pstendpoints},
\begin{equation}\label{eq:endpointPST2}
P_{N-1}(\lambda_k)=e^{-i\phi}e^{it_0\lambda_k},\quad k=0, 1,2,\dots, N-1.
\end{equation} 
Since the $J$ is real, the corresponding OPS is real and therefore 
\begin{equation}\label{eq:realpolyPST}
P_{N-1}(\lambda_k)=\pm 1.
\end{equation} Recall that by proposition \ref{prop:interlacezeros}, the zeros of $P_{N}(x)$ and $P_{N-1}(x)$ interlace and thus \eqref{eq:realpolyPST} is equivalent to 
\begin{equation}
    P_{N-1}(\lambda_k)=(-1)^{N-1+k}, \quad k=0, 1, 2, \dots, N-1.
\end{equation} Hence, \eqref{eq:endpointPST2} implies
\begin{align*}
e^{-i\phi}e^{it_0\lambda_k}&=(-1)^{N-1+k}
\end{align*} and thus
\[
e^{it_0(\lambda_{k}-\lambda_{k-1})}=-1
\] for $k= 1,2,\dots, N-1$. Therefore, 
\[
t_0(\lambda_{k}-\lambda_{k-1})=(2m_k+1)\pi
\] for all $k= 1,2,\dots, N-1$.
%By \eqref{eq:polyPST}, this is equivalent to 
%\begin{equation}
   % e^{-i\phi}e^{it_0\lambda_k}=e^{i\pi k}e^{i\phi + \pi((N-1)+2M_k)}
%\end{equation} where $M_k$ is some integer.
%\Rachel{verify this}
\end{proof}
In \cite{KAY_2010}, it was shown that the above behavior of consecutive eigenvalues along with mirror symmetry (symmetry about the anti-diagonal) of the Jacobi matrix actually characterize the spin chains with nearest neighbor interactions which exhibit PST. 
\begin{theorem}[Kay \cite{KAY_2010}]  For a spin chain on $N$ vertices, PST between endpoints occurs at time $t_0$ if and only if $J$ is mirror symmetric and the eigenvalues satisfy $\lambda_k-\lambda_{k-1}=(2m_k+1)\pi/t_0$.
\end{theorem}
The proof relies on analysis of the eigenvalues and eigenvectors of $J$, however it was shown in \cite{vinet2012} that this characterization can also be stated explicitly in terms of the associated OPS.
\begin{lemma}[Vinet and Zhedanov 2012 \cite{vinet2012}]
Let $J$ be the weighted adjacency matrix of a path graph with $N$ vertices and let $w(x)$ be the weight function of the corresponding OPS. Then $J$ exhibits PST at time $t_0$ if and only if $\lambda_k-\lambda_{k-1}=\frac{(2m_k+1)\pi}{t_0}$ and the orthogonality weights satisfy $w(\lambda_k)=\frac{k_{N-1}}{\|P'_N(\lambda_k)\|}$ where \newline $P_N(x)=(x-\lambda_0)(x-\lambda_1)\dots(x-\lambda_{N-1})$ and $k_{N-1}$ is a normalization constant so that $\sum_{k=0}^{N-1}w(\lambda_k)=1$.
\end{lemma}

This equivalent condition is useful in multiple ways. First, it provides a way to construct a Jacobi matrix with the PST property when given a specific set of eigenvalues. Second, it helps provide a method to obtain new spin chains exhibiting PST from old ones by removing points from the spectrum (see \say{spectral surgery} in \cite{vinet2012}). This is related to the well known Christoffel transformation, which provides conditions when one can modify the orthogonality measure of an OPS and obtain a new OPS with respect to this new measure (see  \cite{BAILEY2023105876} and \cite{Bueno2004DarbouxTA} for information on discrete Darboux transformation). 

Since conditions on the eigenvalues of $J$ are crucial to determining if PST occurs, analyzing the behavior of  eigenvalues has drawn interest. In \cite{Kirkland2019}, Kirkland et al.
state that if a weighted path has no potentials (loops), then the graph must be bipartite. They then then use properties of bipartite graphs to conclude that the eigenvalues of the corresponding Jacobi matrix are symmetric about zero. One may come to the same conclusion bypassing graph theory and using only facts about OPSs. 
\begin{lemma}
    Let $J$ be an $N\times N$ Jacobi matrix with positive off diagonal entries and zeros along the main diagonal. Then its eigenvalues are symmetric about zero.
\end{lemma}
\begin{proof}
Suppose $J$ has zeros along the diagonal and let $\{P_n(x)\}_{n=0}^{N}$ be the corresponding OPS. Then the associated recurrence relation becomes 
\[
xP_n(x)=a_nP_{n+1}(x)+a_{n-1}P_{n-1}(x), \quad n=0,1,2,\dots, N-1.
\] We claim that if $n$ is odd (even), then $P_n(x)$ is an odd (even) function.
Let $n=1$. Then $P_1(x)=a_0^{-1}x$ hence is odd. Similarly, $P_2(x)=a_2^{-1}[xP_1(x)-a_0P_0(x)]$ where $P_0(x)$ and $xP_1(x)$ are both even, hence $P_2(x)$ is even.

Now assume the statement holds for all $n\leq k$. Then $P_{k+1}(x)=xP_{k}(x)-a_{k-1}P_{k-1}(x)$, thus if $k$ is odd, $xP_{k}(x)$ and $P_{k-1}(x)$ are even, thus $P_{k+1}(x)$ is even. A similar argument holds for $k$ even. 

Therefore, if $x_0$ is a zero of $P_k(x)$, then $P_k(-x_0)=\pm P_k(x_0)$ hence $-x_0$ is also a zero of $P_k(x)$. Since the zeros of $P_k(x)$ are real, we have that the zeros of $P_k(x)$ are symmetric about 0. In particular, since the zeros of $P_N(x)$ are the eigenvalues of $J$, we have that the eigenvalues of $J$ are symmetric about zero.
\end{proof}
Kirkland et. al also made the observation that PST  between endpoints, implies PST between interior vertices.
\begin{proposition}[Kirkland et. al \cite{Kirkland2019}]\label{prop:kirkland}
For a weighted path of length $N$ with or without potentials, PST between vertices 1 and $N$ implies PST between vertices $j$ and $N+1-j$ for each $j=2, \dots, N-1$. If for some $2\leq j \leq N-1$, there is PST between vertices $j$ and $N+1-j$, and if in addition none of the eigenvectors of the Hamiltonian has a zero entry in the $j$-th position, then the converse holds. 
\end{proposition}
The proof of the proposition \ref{prop:kirkland} uses Lemma \ref{lemma:PST} and the realness of the OPS corresponding to a real Jacobi matrix. It also uses the fact that $J$ must be mirror symmetric and hence the eigenvectors must be symmetric or antisymmetric, i.e. of the form $(v_k)_j=\pm(v_k)_{N-j+1}, j=1,2,\dots, N$.

\section{Transition Amplitudes}\label{sec:karlinmcgregor}
Recall that PST between vertices means that with probability one, an excitation %(the $|1\rangle$ state) 
at vertex $j$ will transition to vertex $k$. %a state at vertex $j$ will be found at vertex $k$.
Therefore, it is useful to have a formula to compute such transition probabilities. In \cite{Karlin1955}, Karlin and McGregor obtained a seminal formula relating sequences of orthogonal polynomials to transition probabilities of birth and death processes. Since quantum walks are the quantum versions of random walks, it makes sense to study this formula.

Let $\{X_t\}_{t\geq 0}$ be a Markov process on the state space $S=\{0,1,2,\dots\}$ with transition probabilities given by 
\[
P_{i,j}(t)= \mathbb{P}[X_{t+s}=j|X_s=i], \quad i,j \in S.
\] Note that since $\{X_t\}$ is Markov, the transition probabilities do not depend on $s$. Assume that we have nearest-neighbor interactions given by
\begin{align*}
  P_{i, i+1}(t)&=\lambda_it+o(t),\\
    P_{i,i}(t)&=1-(\lambda_i+\mu_i)t+o(t),\\
    P_{i,i-1}(t)&=\mu_it+o(t),\\
    P_{i,j}(t)&=o(t)\quad \text{ for } |i-j|>1  
    \end{align*} as $t \rightarrow 0$.
Then we can associate the transition probabilities with the matrix 
\[
A=\begin{pmatrix}

     -(\lambda_0+\mu_0) & \lambda_0 & 0 & 0 &0&\cdots \\
     \mu_1 & -(\lambda_1+\mu_1) & \lambda_1 & 0 &0& \cdots  \\
     0&\mu_2&-(\lambda_2+\mu_2)&\lambda_2&0&\cdots\\
     \vdots & \vdots& \vdots & \vdots & \vdots &\ddots
\end{pmatrix}.
\] This matrix is a Jacobi matrix with positive upper and lower diagonal entries and thus generates a real OPS. We note that the matrix $A$
satisfies $P'(t)=P(t)A=AP(t)$, and $P(0)=I$.
Karlin and McGregor obtained a formula for computing transition probabilities in terms of the OPS.
\begin{theorem}[Karlin and McGregor \cite{Karlin1955}]\label{thm:karlin}
Let $\{Q_n(x)\}$ be polynomials satisfying 
\begin{align*}
    -xQ_0(x)&= -(\lambda_0+\mu_0)Q_0(x)+\lambda_0Q_1(x)\\
    -xQ_n(x)&=\mu_nQ_{n-1}(x)-(\lambda_n+\mu_n)Q_n(x)+\lambda_nQ_{n+1}(x)\quad n\geq 1  
\end{align*} where $ Q_0(x)\equiv 0$ and $\mu_n>0, n>0$, $\mu_0\geq 0$ and $\lambda_n>0, n \geq 0$. 
Then there exists a positive measure $\psi$ on $[0,\infty)$ such that the $\{Q_n(x)\}$ are orthogonal with respect to $\psi$ and 
the transition probabilities $P_{i,j}(t)$ may be represented by
%the matrix $P(t)=(P_{i,j}(t))$ defined by
\begin{equation}\label{eq:BDKarlinMcGregor}
P_{i,j}(t)=\frac{\displaystyle \int_{0}^{\infty}e^{-xt}Q_i(x)Q_j(x)d\psi(x)}{\displaystyle \int
_0^{\infty}Q_j(x)^2d\psi(x)}.
\end{equation}
%satisfies $P'(t)=AP(t)$, $P'(t)=P(t)A$ and $P(0)=I$.
\end{theorem}

\begin{remark}
    We note here that a remarkable part of the above theorem is the connection between birth and death processes and the Steiltjes moment problem.
\end{remark}%\Rachel{In the case where $A$ is finite, $\psi$ will be have finite support.}
\noindent The quantum analogue of a transition probability is called a \textit{transition amplitude.}

\begin{definition}
 The \textit{transition amplitude}, $c_{nm}(t)$, for a state at vertex $n$ %$|j\rangle$ 
 to be found at vertex $m$ %state |k\rangle$ 
  at time $t$, and is given by
   \[
   c_{nm}(t)={\bf e}_m^{\top}e^{itJ}{\bf e}_n, \quad n,m=0,1,2,\dots, N-1.
   \]
\end{definition}
The probability of a state at vertex $n$ to be found at vertex $m$ is given by $|c_{nm}(t)|^2$.
Because $J$ may be diagonalized using the corresponding OPS, it is quick to derive the following formula for transition amplitudes. 
\begin{theorem}
Let $G$ be a path graph on $N$ vertices and let $\{P_n(x)\}_{n=0}^N$ be the orthonormal OPS corresponding to the weighted adjacency matrix $J$ with weight function $w(x)$. Let $\{\lambda_k\}_{k=0}^{N-1}$ be the zeros of $P_N(x)$.  %Define $\phi_n(x)=\sqrt{w(x)}P_n(x)$. 
Then the transition amplitudes are given by
\begin{equation}
    c_{nm}(t)=\sum_{k=0}^{N-1}P_n(\lambda_k)P_m(\lambda_k)w(\lambda_k)e^{i\lambda_kt}, \quad n,m = 0,1,2\dots, N-1.
\end{equation}

\end{theorem}

We illustrate this idea with an explicit example utilizing the Krawtchouk polynomials. %$\tilde{v_k}=\sqrt{w(\lambda_k)}v_k$

  \begin{example} 
   Consider a spin chain with $5$ qubits and let the Hamiltonian be represented by the following Jacobi matrix:
   \[
J=\begin{pmatrix}
2 & 1&0 & 0&0   \\
     1 & 2 & \sqrt{3/2} &0&0   \\
     0 & \sqrt{3/2} & 2 & \sqrt{3/2}&0 \\
      0&0 & \sqrt{3/2} & 2&1\\
      0&0&0&1&2
\end{pmatrix}.
\]
 Then $J$ generates the orthonormal Krawtchouk polynomials \[\tilde{K}_n(x;1/2,4):=\frac{K_n(x;1/2,4)}{a_0a_1\dots a_{n-1}}\] where $a_n=\sqrt{\frac{(4-n)(n+1)}{4}}$ and $K_n(x;1/2,4)$ is defined as in \eqref{def:krawtchouk}. Recall from Example \ref{ex:krawtchouk} that $K_5(x;1/2,4)$ has zeros at $x=0,1,2,3,4$, and thus so does $\tilde{K}(x;1/2,4)$, hence the eigenvalues of $J$ are $0,1,2,3$ and $4$.
% We note that the $\tilde{K}_n(x;4,1/2)$ satisfy
% \[\frac{\sqrt{(4-n)(n+1)}}{4}\tilde{K}_{n+1}(x)+2\tilde{K}_n(x)+\frac{\sqrt{(4-n+1)n}}{2}\tilde{K}_{n-1}(x)=x\hat{P}_n(x)
 %\] for $n=0,1,2,3,4$.

The corresponding orthogonality weight function is $w(x)=\frac{1}{2^4}\binom{4}{x}$ and thus the transition amplitudes
%\tilde{K}_n(x)=\frac{K_n(x)}{\lambda_0\lambda_2\dots\lambda_{n-1}} where $\lambda_n=\frac{(M+1-n)n}{4}$-check
of this system are given by
\[
c_{nm}(t)=\sum_{x=0}^4 \frac{1}{2^4}\binom{4}{x}\tilde{K}_n(x)\tilde{K}_m(x)e^{ixt}
\] %where $\tilde{K}_n(x)=\frac{K_n(x)}{||K_n(x)||}$. 
Using the fact that $\tilde{K}_0(x)=1$ and $\tilde{K}_4(x)=\frac{2}{3}\left(x^4-8x^3+20x^2-16x+3/2\right)$, one can see that $c_{04}(\pi)=1$ and thus the quantum system corresponding to the Krawtchouk polynomials exhibits PST at $t=\pi$. This is consistent with the fact that $J$ is mirror-symmetric and the eigenvalues satisfy $\left(\lambda_k-\lambda_{k-1}\right)\pi=\pi$.
\begin{remark}\label{remark:krawtchoukPST}
It is not too difficult to see that this result holds for all $N+1\times N+1$ Jacobi matrices corresponding to the orthonormal Krawtchouk polynomials with $M=N$ and $p=1/2$.
\end{remark}
\end{example}
%Although the above theory shows how to produce systems which exhibit PST, perfect conditions are assumed which is not the case in practicality. To tackle PST when there are sources of error, the concept of "pretty good state transfer (PGST)" (also referred to as "almost perfect state transfer") has been studied. We do not go into details of PGST in this review, but refer readers to  \cite{Godsil2017},  \cite{Godsil_2012} and \cite{Vinet_APST} for further discussion. 

   \section{Exceptional Orthogonal Polynomials}\label{sec:XOPS}
There has been recent activity in the study of Exceptional Orthogonal Polynomials (XOPs) which are generalizations of classical orthogonal polynomials but the sequences are missing finitely many degrees (see \cite{castro2022new}, \cite{D21}, \cite{MR4367467}, or
\cite{KrawtchoukVinet} for recent discussion and applications of XOPs). One of the properties that XOPs share with OPSs is that they also satisfy recurrence relations, however, these are of order $n \geq 4$, and thus the corresponding operator is represented by an $n$-diagonal matrix with $n \geq 4$. Therefore, it is natural to ask if there is a relationship between XOPs and quantum walks on graphs which exhibit interactions with more than just nearest-neighbors. In \cite{Miki_2022}, the authors explore an explicit application of X-Krawtchouk polynomials to this larger class of graphs, which we recall below. 
\begin{definition}
 Let $N$ be a positive integer and let $S=\{0, 1, 2, \dots, N, N+3\}$. Then define the \textit{$X_2$- Krawtchouk Polynomials} $\{\hat{K}^{(2)}_n(x;p,N)\}_{n \in S}$  by
 \[
 \hat{K}^{(2)}_n(x;p,N)= \frac{1}{N+3-n}\begin{vmatrix}
     (N-x)K_2(x-N+1;p,-N-2)(x) & K_n(x)\\
     -(1-x)K_2(x-N+2;p, -N-2)(x) & K_n(x+1) 
 \end{vmatrix}
 \] for $n=0,1,2, \dots N$ and 
\begin{multline}
\displaystyle
     \hat{K}^{(2)}_{N+3}(x;p,N)=\sum_{0 \leq k,j\leq 2} \frac{(-2)_j(-2)_k(p-1)^{2-k}p^{2-j}}{j!k!}(-N-3)_{2-j}(-N-3)_{2-k}\\ 
\times K_{N+k+j+1}(x+k+1; p, N+k+j+1)(x).
 \end{multline}

\end{definition} 

\noindent The $X_2$-Krawtchouk polynomials form an orthogonal sequence with respect to 
 \[
 \sum_{x=-1}^N \hat{K}^{(2)}_n(x;p,N)\hat{K}^{(2)}_m(x;p,N)\hat{w}(x)
= \hat{h}_n\delta_{n,m} \] where 
\begin{equation}
\hat{w}(x)=\frac{w(x+1; p, N+1)}{K_2(x-N-1;p,-N-2)K_2(x-N;p,-N-2)}
\end{equation} where $w(x)= \binom{N}{x} p^x(1-p)^{N-x}$ is the weight function of the ordinary Krawtchouk polynomials defined in \eqref{def:Krawtchouk} and  and $\hat{h}_n$ are constants. 

The $X$-Krawtchouk polynomials are formed by performing Darboux transformation to the ordinary Krawtchouk difference operator $L$ which is given by
\[
L[f(x)]:=p(N-x)\left[ f(x+1)-f(x)\right]+x(1-p)\left[f(x-1)-f(x) \right].
\]
They satisfy a 7-term recurrence relation,
\begin{equation}
\lambda_x\hat{K}^{(2)}_n(x;p,N)=\sum_{k=-3}^{3}\beta_{n,k}\hat{K}^{(2)}_{n+k}(x;p,N),\quad x \in \{-1, 0,1,2,\dots, N\}
\end{equation}
where 
\[\lambda_x=-K_3(x-N;p,-N-1)
\] is given in terms of the ordinary Krawtchouk polynomial of degree 3. Formulas for the coefficients $\beta_{n,k}$ may be found in \cite{Miki_2022} but we note that 
$\beta_{N+3,k}=0$ for $k=-2,-2,1,2$ and $3$. Thus, one may associate the $X_2$-Krawtchouk polynomials with a quantum walk on $N+2$  vertices, labeled $0,1,2,\dots, N, N+3$, where transition occurs amongst the three nearest-neighbors except for vertex $N+3$ which only interacts with vertex $N$.

\begin{theorem}[Miki, Tsujimoto and Vinet \cite{Miki_2022}]\label{thrm:KarlinXKrawtchouk}
The transition amplitude of a quantum walk with Jacobi matrix corresponding to the $X_2$-Krawtchouk polynomials is given by
\begin{equation}
    c_{nm}(t)=\sum_{x=-1}^NT_n(x;p)T_m(x;p)e^{-i\lambda_xt}
\end{equation}
\end{theorem}
The proof uses the fact that the matrix associated with the orthonormal $X_2$-Krawtchouk polynomials, $\tilde{K}^{(2)}_n(x;p,N)$, is Hermitian, and thus may be diagonalized using the eigenvectors $v_x=(T_1(x;p), T_2(x;,p),\dots, T_{N+2}(x;p))^T$, which are given in terms of the $X_2$-Krawtchouk functions:
\[
 T_n(x;p)=
\begin{cases}
   \sqrt{\hat{w}(x)}\tilde{K}^{(2)}_{n-1}(x;p,N) \quad & n=1,2,\dots,N+1\\
  \sqrt{\hat{w}(x)}\tilde{K}^{(2)}_{N+3}(x;p,N) & n=N+2.
\end{cases}
\]\\
As a corollary to Theorem \ref{thrm:KarlinXKrawtchouk}, it can be shown that for $p=1/2$, perfect return occurs with probability one, but there is no PST between the end vertices. This is contrary to the model associated to the ordinary Krawtchouk polynomials as stated in Remark \ref{remark:krawtchoukPST}.

\section{Open Questions}
It is shown in \cite{kempton2016} that given an unweighted path graph of length $n \geq 4$, there does not exist a potential which induces PST between the endpoints. It has also been shown that unweighted complete graphs with $n\geq 3$ vertices do not have PST. Tamon and Kendon provide a discussion  in \cite{tamon2011} of results regarding PST for complete graphs, path graphs, hypercubic graphs, Hamming graphs and integral circulants, but this list is not finished, thus the search to characterize all graphs with PST remains.

 As mentioned in Section \ref{sec:PST},  the authors of \cite{vinet2012} describe a process called \say{spectral surgery} which allows one to construct new Jacobi matrices from old ones exhibiting PST so that the new system also exhibits PST. The procedure involves removing points from the spectrum and looking at new weight functions of the form $\tilde{w}(x)=\prod_{k=j}^n(x-\lambda_k)w(x)$, with $0\leq j,n\leq N-1$. %with $ 0\leq j,k\leq N-1.$ 
 However, there is a restriction that the spectral points must be chosen so that the new weights are positive. Therefore, investigation into performing spectral surgery on arbitrary points of the spectrum is warranted. We remark that this corresponds to performing the Christoffel transformation on a family of orthogonal polynomials at points within the support of the orthogonality measure.

In reference to \cite{Miki_2022}, a possible direction is to come up with more explicit examples of quantum walks on graphs  related to other families of XOPs. The authors suggest looking at the  dual Hahn polynomials which is another example of a finite OPS. It could also be useful to study how to extend the applications of XOPs to higher dimensional graphs. Thus, exploration of multivariate XOPs is a challenging but exciting path to take.

Recall Proposition \ref{prop:kirkland}, which states that PST between end vertices implies PST between opposite interior vertices. The converse to the statement is only proven for graphs with Hamiltonians which have eigenvectors of a specific type. Thus, it remains an open question as to if the converse holds for arbitrary Jacobi matrices.
%In the same paper, the authors conjecture that for XX dynamics on a path of length $n \geq 4$ with or without loops, if PST occurs at time $t=\pi$ then there must be at least one irrational weight. While they prove the result for $n=4$, $n\equiv 5 \mod{8}$ and $n\equiv 3 \mod{8}, n \neq 3$, it still remains to be shown for all $n \geq 4$.

Finally, in \cite{bailey2024}, the authors study Jacobi matrices which exhibit PST at time $T_0$ as well as have the property that there exists some time $t<T_0$ such that the first component of $e^{itJ}{\bf{e_0}}$ is 0, while the last component is not 1. The authors call this phenomena \textit{Early State Exclusion} (ESE) and show the existence of an $N \times N$ Jacobi matrix with ESE for any even $N \geq 4$. The case for odd $N \geq 5$ is still to be determined. 

While this is far from a complete list of open questions regarding PST in quantum walks on graphs, we hope that it demonstrates the many directions of the topic and possibly inspires readers to explore further paths. \\

\noindent {\bf Acknowledgments.} I am extremely grateful to Prof. Michael Nathanson for the helpful discussions and comments which have strengthened this paper.

%\nocite{*} % prints full bibliography
\bibliography{bib.bib}
\bibliographystyle{plain}
\end{document}